\documentclass{ieeetran}
\usepackage{amsmath}
\usepackage{amssymb}
\usepackage{bm}
\usepackage{mathrsfs}
\usepackage{graphicx}

\newtheorem{theorem}{Theorem}
\newtheorem{proposition}{Proposition}
\newtheorem{corollary}{Corollary}
\newtheorem{example}{Example}

\renewcommand{\phi}{\varphi}

\newcommand{\diag}{\operatorname{diag}}

\newcommand{\bx}{\boldsymbol{x}}

\newcommand{\bZ}{\boldsymbol{Z}}

\newcommand{\bz}{\boldsymbol{z}}

\newcommand{\bP}{\boldsymbol{P}}

\newcommand{\bD}{\boldsymbol{D}}
\newcommand{\sX}{\mathcal{X}}
\newcommand{\sY}{\mathcal{Y}}
\newcommand{\bSigma}{\boldsymbol\Sigma}

\def\reals{\mathbb{R}}
\def\bx{\boldsymbol{x}}
\def\b0{\mathbf{0}}
\def\bP{\boldsymbol{P}}

\def\bSigma{\boldsymbol\Sigma}

\def\bC{\boldsymbol{C}}

\def\bA{\boldsymbol{A}}
\def\bI{\mathbf{I}}

\def\rmV{\mathrm{V}}

\def\PSD{\mathrm{S_{++}}}

\usepackage{algorithmic}
\usepackage{algorithm}
\newcommand{\di}{{\,\mathrm{d}}}

\renewcommand{\(}{\left(}
\renewcommand{\)}{\right)}

\title{Multivariate Generalized Gaussian Distribution: Convexity and Graphical Models}

\author{Teng Zhang, Ami Wiesel and Maria Sabrina Greco\thanks{Copyright (c) 2012 IEEE. Personal use of this material is permitted. However, permission to use this material for any other purposes must be obtained from the IEEE by sending a request to pubs-permissions@ieee.org.

Teng Zhang is with the University of Minnesota, Institute for Mathematics and its Applications, USA. Ami Wiesel is with The Hebrew University of Jerusalem, Israel. Maria Sabrina Greco is with the Universita di Pisa, Italy. Emails: zhang620@umn.edu, amiw@cs.huji.ac.il, and m.greco@iet.unipi.it. A. Wiesel was partially supported by the Intel Collaboration Research Institute for Computational Intelligence.}}

\begin{document}
\maketitle

\begin{abstract}
We consider covariance estimation in the multivariate
generalized Gaussian distribution (MGGD) and elliptically symmetric (ES) distribution. The maximum likelihood
optimization associated with this problem is non-convex, yet it
has been proved that its global solution can be often
computed via simple fixed point iterations. Our first contribution is a new analysis of this likelihood based on geodesic convexity that requires weaker assumptions. Our second contribution is a generalized framework for structured covariance estimation under sparsity constraints. We show that the optimizations can be formulated as convex minimization as long the MGGD shape parameter is larger than half and  the sparsity pattern is chordal. These include, for example, maximum likelihood estimation of banded inverse covariances in multivariate Laplace distributions, which are associated with time varying autoregressive processes.
\end{abstract}

\section{Introduction}

Covariance estimation  is a fundamental problem in multivariate statistics. Many techniques for hypothesis testing, inference, denoising and prediction rely on accurate estimation of the true covariance. The problem is challenging when the available data is high dimensional and non-Gaussian. Such settings are typical in many applications including speech, radar, wireless communication, finance and more. These led to a growing interest in both robust and structured covariance estimation. Specifically, in this paper, we consider maximum likelihood estimation (MLE) in the multivariate generalized Gaussian distribution, with and without sparsity constraints on the inverse covariance.

The first part of this paper considers the geodesic convexity in MLE in elliptically symmetric (ES) distributions. Methods for robust covariance estimation date back to the early works of \cite{huberrobust,Tyler1987}. A popular approach is Tyler's scatter estimate. It involves a non-convex optimization yet can be solved via simple fixed point iteration. It has been rigorously analyzed and successfully applied to different problems \cite{Tyler1987,pascal2008covariance}. Recently, it was shown that the result is in fact the solution to a geodesically convex minimization \cite{Auderset2005,Chen2011,Wiesel2012LSE,Wiesel2012,zhang2012robust}. Geodesic convexity ensures that any local minima is also globally optimal and leads to a much simpler analysis \cite{Rapcsak1991}. It also allows for numerous extensions, e.g., regularized solutions. In a competing line of works, a different class of robust covariance estimation techniques was proposed based on the MGGD \cite{Novey2010,Chmielewski1981}. A well known example of MGGD is the multivariate Laplace distribution \cite{eltoft2006multivariate}. Fixed point iterations for MGGD estimation and their analyses has recently been considered in \cite{Ollila2012,bombrun2012performance,Pascal2012}. The first contribution in this paper is a new analysis which shows that the negative log-likelihood in MGGD is also geodesically convex. This result requires weaker conditions than  previous analyses, provides more intuition and paves the road to numerous generalizations.

The second part of this paper addresses structured covariance estimation. Structure exploitation is a main ingredient in modern statistics that allows accurate high dimensional estimation via a small number of samples. A promising approach is based on sparse inverse covariance models. In the multivariate Gaussian case these are known as graphical models and characterize conditional independence \cite{dempster1972covariance,lauritzen1996graphical,Bickel2008,andersen1995linear}. These models have been successfully applied to  speech recognition, sensor networks, computer networks and other fields in signal processing  \cite{Chen2011,wiesel2009decomposable}. Our goal is to combine such models with non-Gaussian distributions, e.g., MGGD. Recently, a similar problem was addressed using an expectation maximization technique \cite{finegold2011robust}. Another line of work focused on combining Tyler's scatter estimate with a banded inverse covariance prior \cite{abramovich2012expected,Abramovich2012}. Graphical models for transelliptical distributions were discussed in \cite{liu2012transelliptical}. In this paper, we combine the MGGD framework with prior sparsity constraints on the inverse covariance. We show that the optimization can be formulated into a convex form as long as the MGGD scale parameter is larger than half and the sparsity satisfies a chordal structure. Chordal models, also known as decomposable or triangulated models, include banded structures, multiscale settings and other practical scenarios \cite{Wiesel_decomposable2010,Fukuda99exploitingsparsity,Dahl2008}. Such structures are associated with a perfect ordering of the variables. A typical example is banded models associated with time-varying autoregressive processes \cite{abramovich2007time}.

Our results are also applicable to the case of unknown sparsity pattern, i.e., structure learning via sparsity inducing penalties, but require prior knowledge of the perfect order. Recent works on structure learning in directed acyclic graphs provide data driven techniques for learning this order \cite{rutimann2009high,rolfs2012natural}.

The outline of the paper is as follows. We begin in Section \ref{prel} with a few mathematical preliminaries that will be useful in our work. Then, we continue to our two main contributions. In Section \ref{gcovnex_mggd} we provide a new geodesic analysis of MGGD estimation, and in Section \ref{convex_mggd} we introduce a convex optimization framework for chordal structured MGGD estimation. Simulation results are described in Section \ref{simulations}, and concluding remarks are offered in Section \ref{conclusion}.

We use the following notations. We denote the set of real, symmetric and positive definite matrices by $\PSD(p)\subset\mathbb{R}^{p\times p}$. We denote the span operator by $\mathrm{sp}\{\cdot\}$.
%A $d$-dimensional linear subspace consists all points in the set $\{\bz\in\mathbb{C}^p:\bz=\sum_{i=1}^d\lambda_i\bz_i, \,\,\lambda_i\in\mathbb{C}\}$, where $\bz_1,\bz_2,\cdots,\bz_d$ are $d$ independent vectors in $\mathbb{C}^p$. Though intuitively it is a $2d$-dimensional linear subspace in $\reals^{2p}$, we define its dimension to be $d$ in this paper. DEFINE DIM

\section{Preliminaries}\label{prel}
We begin with a brief review of two mathematical concepts which will be instrumental in the next sections.

\subsection{Geodesic convexity}
Geodesic convexity is an extension of classical convexity which replaces lines with geodesic paths in manifolds. More details on this topic can be found in \cite{Rapcsak1991}. Given a Riemannian manifold $\mathcal{M}$ and a set $\mathcal{A}\subset \mathcal{M}$, we say a function $f: \mathcal{A}\rightarrow \reals$ is geodesically convex, if every geodesic $\gamma_{xy}$ of $\mathcal{M}$ with endpoints $x,y\in\mathcal{A}$ (i.e., $\gamma_{xy}$ is a function from $[0,1]$ to $\mathcal{M}$ with $\gamma_{xy}(0)=x$ and $\gamma_{xy}(1)=y$) lies in $\mathcal{A}$, and
  \begin{eqnarray}
  &  \text{$f\(\gamma_{xy}(t)\)\leq (1-t)f(x)+t f(y)$}\nonumber\\
  &\text{ for any $x,y\in\mathcal{A}$ and $0< t< 1$}\label{eq:convex_defi}.\end{eqnarray}
If the inequality in~\eqref{eq:convex_defi} is replaced by strict inequality, we call the function $f$ geodesically strictly convex. An equivalent definition follows  from \cite[Theorem 1.1.4]{niculescu2006convex}.
\begin{proposition}\label{prop:mean}
For continuous function $f$, the definition in \eqref{eq:convex_defi} is equivalent to the condition
  \begin{eqnarray}
  &\text{$f\(\gamma_{xy}\(\frac{1}{2}\)\)\leq \frac{1}{2}f(x)+\frac{1}{2} f(y)$}\nonumber\\
  &\text{for any $x,y\in\mathcal{A}$}\label{eq:convex_defi1}.\end{eqnarray}
\end{proposition}

The importance of geodesic convexity stems from the following properties (see \cite[Theorem 2.1]{Rapcsak1991} for more details).
\begin{proposition}
Any local minimizer of a geodesically convex function is also its global minimizer.
\end{proposition}
\begin{proposition}
Any strictly geodesically convex function has a unique global minimizer.
\end{proposition}

In particular, we consider geodesic convexity on  the manifold of positive definite matrices denoted by $\PSD(p)$. The geodesic connecting $\bSigma_1\in\PSD(p)$ and $\bSigma_2\in\PSD(p)$ is defined as~\cite[Chapter 6]{bhatia2007positive}
\begin{eqnarray}\label{eq:geodesic}
\gamma_{\bSigma_1\bSigma_2}(t)=\bSigma_1^{\frac{1}{2}}\(\bSigma_1^{-\frac{1}{2}}\bSigma_2\bSigma_1^{-\frac{1}{2}}\)^t\bSigma_1^{\frac{1}{2}}.
\end{eqnarray}

Geodesic convexity of Tyler's likelihood has been identified in \cite{Auderset2005,Wiesel2012LSE,Wiesel2012,zhang2012robust}. In Section \ref{gcovnex_mggd}, we continue this line of works and show that the MGGD likelihood is also geodesically convex.

\subsection{Chordal graphs}
In statistical graphical modeling, graphs are used to characterize the sparsity pattern of the corresponding inverse covariance matrices. When the pattern belongs to a special class known as chordal graphs, these concentration matrices satisfy an appealing structure which will be exploited in Section \ref{convex_mggd}. More details on chordal graphs and their relation to graphical models can be found in \cite{lauritzen1996graphical,Wiesel_decomposable2010,Fukuda99exploitingsparsity,Dahl2008}.

A graph $G(V,E)$ is chordal if every cycle of length $\geq 4$ has an edge joining two nonconsecutive vertices of the cycle. For a chordal graph, there is a perfect elimination ordering of vertices, $(v_1, v_2, \cdots, v_n)$ such that for any $1\leq i\leq n$, the neighbor of $v_i$, $\mathrm{Adj}(v_i)=(u\in V: (u,v_i)\in E)$, satisfies that  $\mathrm{Adj}(v_i)\cap \{v_{i+1},v_{i+2},\cdots,v_{n}\}$ induces a fully connected clique, i.e., a set of fully connected nodes.

It is convenient to define the sparsity pattern of a square $n\times n$ matrix $\bC$ via a graph $G(V,E)$ with $n$ vertices.  We say that $\bC$ is $G$-sparse if
\begin{eqnarray}
 [\bC]_{ij}=0,\qquad \text{for all} \quad (i,j)\notin E.
\end{eqnarray}

%By applying the reverse order of the perfect elimination order, the Cholesky decomposition of $\Sigma$ has the same sparsity pattern as ~\cite{Fukuda99exploitingsparsity}.

A Cholesky decomposition is a generalization of the squared root operation to positive definite matrices. Any $\bSigma\in\PSD(p)$ has a unique decomposition $\bSigma=\bC\bC^T$ where the Cholesky factor $\bC$ is a lower-triangular matrix with positive diagonal elements.

The following result characterizes the relation between sparse Cholesky decompositions and chordal graphs.\begin{proposition}\label{prop:chol}
Let $G(V,E)$ be a chordal graph with a natural\footnote{If the order is not natural, this result holds by permuting the columns of the matrix.} perfect order, i.e., $v_i=i$ for $i=1,\cdots,n$. If $\bSigma$ is $G$-sparse, real and positive definite, then there exists a unique $G$-sparse lower triangular $\bC$ with positive definite diagonal entries such that $\bSigma=\bC\bC^T$. If $\bC$ is $G$-sparse and lower-triangular, then $\bC\bC^T$ is $G$-sparse and positive definite.
\end{proposition}
The existence proof can be found in~\cite[Section 2.1]{Fukuda99exploitingsparsity}, and the uniqueness is a direct consequence of the perfect elimination order.

\begin{example}
For a banded matrix with band width $d$, it corresponds to a graph $G(V,E)$ defined as follows: $(i,j)\in E$ when $|i-j|< d$. Now we check that this is a chordal graph: for every cycle with length $\geq 4$, the indices of the four nodes differ by $d$ at most, therefore the edges connect all nodes in the cycle, which corresponds to the definition of chordal graph.

The elimination order for this chordal graph turns out to be the natural order $v_i=i$. Therefore Proposition~\ref{prop:chol} shows that the Cholesky decomposition of a banded matrix $\bSigma$ is the product of a banded (with the same bandwidth $d$) lower-triangular matrix with its transpose.

One can also easily show that the product of a banded lower-triangular matrix with its transpose is still a banded (with the same band width), positive definite matrix, which exemplifies the last sentence in Proposition~\ref{prop:chol}.
\end{example}
For more examples of Chordal graph and its associated perfect order, we recommend the examples and graphs in~\cite[Section II.A]{wiesel2009decomposable} and ~\cite[Section 3]{Dahl2008}.
%The following result characterizes the relation between sparse Cholesky decompositions and chordal graphs.
%\begin{proposition}\label{prop:chol}
%(a) If $\bSigma$ is sparse according to $G(V,E)$, and the columns of $\bSigma$ follow the perfect elimination order, then its Cholesky factor $\bC$ is also sparse according to $G(V,E)$.
%(b) If $\bC$ is a lower-triangular matrix with positive diagonal elements, and sparse according to $G(V,E)$, then $\bC\bC^T$ is also sparse according to $G(V,E)$.
%\end{proposition}
%The proof of part (a) can be found in~\cite[Section 2.1]{Fukuda99exploitingsparsity}, and the proof of part (b) is a direct consequence of the definition of perfect elimination order. TENG - CANT YOU MAKE THE PROP MORE PRETTY???

\section{Geodesic convexity in MGGD}\label{gcovnex_mggd}

In this section, we consider unconstrained MLE in MGGD. More precisely, we address a more general family of  elliptically symmetric (ES) distributions (see \cite{Chmielewski1981} for a review and \cite{Ollila2012} for a recent generalization to complex case). These problems involve non-convex minimizations, yet it has been shown that their global minima can be efficiently found using simple fixed point iterations. We will show that the negative log-likelihoods are in fact geodesically convex, and that this may be the underlying principle behind their success.

%The CES distribution can be considered as the generalization of real elliptically symmetric distribution~\cite{Chmielewski1981} to complex space. The definition is as follows~\cite[(16)]{Ollila2012}:
An random variable $\bz\in\mathbb{R}^p$ has a ES distribution in real space if its probability density function (p.d.f.) is
\begin{equation}\label{eq:ces_def}
f(\bz)=C_{p,g}|\bSigma|^{-0.5}g\(\(\bz-\mu\)^T\bSigma^{-1}\(\bz-\mu\)\),
\end{equation}
where $g: [0,\infty)\rightarrow (0,\infty)$ such that $\int_0^\infty t^{p-1}g(t^2)\di t<\infty$,  $C_{p,g}$ is a normalization constant such that the integral of the distribution is $1$.
In~\eqref{eq:ces_def}, $\bSigma\in \PSD(p)$ is called a scatter matrix, and $\mu\in\mathbb{R}^p$ is the center of the distribution.
%This family of distribution is parameterized by
%\begin{equation}
%p(\bz;\bSigma)=C(\rho)|\bSigma|^{-\frac{1}{2}}\exp(-\rho(\bz\bSigma^{-1}\bz)),
%\end{equation}
%where $C(\rho)$ is a normalization constant such that the integral of the distribution is $1$.
%
%The probability density function of a $p$-dimensional
%multivariate generalized Gaussian distribution (MGGD) is defined as
%\begin{equation}
%p(\bz;\bSigma)=\frac{\beta\Gamma(\frac{p}{2})}{\pi^{\frac{p}{2}}\Gamma(\frac{p}{2\beta})2^{\frac{p}{2\beta}}|\bSigma|^{\frac{1}{2}}}\exp(-\frac{1}{2}(\bz^T\bSigma^{-1}\bz)^\beta).
%\end{equation}

MGGD~\cite{Novey2010} is a widely used special case of ES when
\begin{eqnarray}
 g(x)=\exp\(-x^\beta/2\)
\end{eqnarray}
where $\beta$ is the shape parameter. In particular, for $\beta = 0.5$ it gives multivariate analog of Laplace distribution, and the multivariate Gaussian distribution is obtained for $\beta=1$. %In theoretical sections we will focus on CES distribution since it is more general than MGGD; and in experiment section we will focus on MGGD due to its simplicity.

We consider the estimation of $\bSigma$ given $n$ independent and identically distributed (i.i.d) realizations of a zero mean ES random vector denoted by $\bz_1,\cdots,\bz_n$. The MLE is the parameter that minimizes the negative-log-likelihood
\[
L_0(\bSigma)=\sum_{i=1}^{n}\rho\(\bz_i^T\bSigma^{-1}\bz_i\)+\frac{n}{2}\log\det(\bSigma),
\]
where $\rho(x)=-\log g(x)$. The following Theorem~\ref{thm:mainn} characterizes the existence and uniqueness properties of this MLE. Theorem~\ref{thm:mainn}(a) characterizes the uniqueness and is the main contribution in this section;  Theorem~\ref{thm:mainn}(b) characterize the existence and is borrowed from~\cite[Theorem 2.1]{Tyler1987}.
% We give conditions for the uniqueness and existence of the minimizer in Theorem~\ref{thm:mainn} from two respects. First, we show the geodesic convexity of the objective function $L_0(\bSigma)$ and therefore the uniqueness of the minimizer in Theorem~\ref{thm:mainn}(a), since convexity means that any of its local minimize is its global minimizer. Second, we give the condition for the existence of the minimizer in Theorem~\ref{thm:mainn}(b). While the uniqueness and existence of this estimator has been discussed in~\cite[Section V.A]{Ollila2012}, our theorem provides less restrictive conditions.
 \begin{theorem}\label{thm:mainn}
(a) Assume that $\rho(x)$ is continuous in $(0,\infty)$, nondecreasing and $\rho(e^x)$ is convex,  then $L_0(\bSigma)$ is geodesically convex in $\PSD(p)$. If additionally $\mathrm{sp}\{\bz_1,\bz_2,\cdots,\bz_n\}=\mathbb{R}^p$, $\rho(x)$ is strictly increasing and $\rho(e^x)$ is strictly convex, then $L_0(\bSigma)$ is geodesically strictly convex in $\PSD(p)$.

(b) If $\rho(x)$ is continuous in $[0,\infty)$, $a_1=\sup\{a|x^{a/2}\exp(-\rho(x))\rightarrow 0\,\,\,\, \text{as $x\rightarrow \infty$}\}$ is positive, and  $\sX=\{\bz_i\}_{i=1}^n$, \begin{equation}\label{eq:existence} \text{ $\frac{|\sX\cap\rmV|}{n} < 1 - \frac{p-\dim(\rmV)}{a_1}$ for any linear subspace $\rmV\in\mathbb{R}^p$,}\end{equation}then there exists a minimizer of $L_0(\bSigma)$.
 \end{theorem}

Before proving this theorem, a few comments are in order. We remark that the condition $\mathrm{sp}\{\bz_1,\bz_2,\cdots,\bz_n\}=\mathbb{R}^p$ is also implicated by~\eqref{eq:existence}, since otherwise  $\rmV=\mathrm{sp}\{\bz_1,\bz_2,\cdots,\bz_n\}$ violated the assumption. This condition is necessary since otherwise $L_0(\bP_{\rmV})\rightarrow -\infty$ as $\bSigma\rightarrow \bP_{\rmV}$, where $\bP_{\rmV}$ is the projector to subspace $\rmV$) and therefore the minimizer of $L_0(\bSigma)$ does not exist.

 Theorem \ref{thm:mainn} relaxes the conditions for the uniqueness/existence of the minimizer of $L_0(\bSigma)$ presented in~\cite[Theorem 2.2]{Kent1991}:\\
% ~\cite[Section V.A]{Ollila2012}:\\
 \textbf{M1} $\rho'(x)$ is non-negative, continuous and nonincreasing.\\
 \textbf{M2} $x\rho'(x)$ is strictly increasing.\\
 \textbf{M3} condition~\eqref{eq:existence}.\\
% \textbf{M3} Represent each complex data vector $\bz_i=\bx_i+i \by_i$ by two vectors in $R^{2p}$: $\bv_i=(\bx_i^T,\by_i^T)^T$ and $\bv_{n+i}=(-\by_i^T,\bx_i^T)^T$, and define $\mathcal{V}:=\{\bv_1,\bv_2\cdots,\bv_{2n}\}$. Then for all linear subspaces $\rmV\subset\reals^{2p}$, \[
% \frac{|\mathcal{V}\cap\rmV|}{2n}<1-\frac{2p-\dim(V)}{2a_1}.
% \]
Our assumption that $\rho$ is increasing corresponds to $\rho'(x)>0$, which follows from \textbf{M1}, where $\rho'$ is nonnegative and \textbf{M2} (which exclude the possibility that $\rho'(x)=0$ for $x>0$); and our assumption that $\rho(e^x)$ is strictly convex corresponds to \textbf{M2}, where $x\rho'(x)$ is strictly increasing. In comparison, our conditions does not require $\rho$ to be differentiable, and we do not assume their assumptions in  \textbf{M1} that $\rho'(x)$ is continuous or nonincreasing.

%Our condition in~\eqref{eq:existence} corresponds to the assumption  \textbf{M3}. We remark that it is also slightly less restrictive than the assumption  \textbf{M3}. One can show it by mapping the $d$-dimensional subspace $\rmV$~\eqref{eq:existence} to the $2d$-dimensional subspace $\rmV$ in  \textbf{M3} by the mapping from $\mathbb{Z}^p$ to $\reals^{2p}$ described in assumption  \textbf{M3}.

A special case of Theorem~\ref{thm:mainn} is Tyler's M-estimator in which $\rho(x)=\log(x)/2p$. Following the first part of Theorem~\ref{thm:mainn}(a), $L_0(\bSigma)$ is geodesically convex, which has been previously identified in \cite{Auderset2005,Wiesel2012LSE,Wiesel2012,zhang2012robust}. The geodesic convexity does not contradict the non-uniqueness of Tyler's M-estimator since this convexity is not strict.  Another special case is the class of MGGD estimators. The following corollary follows from Theorem~\ref{thm:mainn} and the fact that $a_1$ in Theorem~\ref{thm:mainn}(b) is $\infty$.
 \begin{corollary}
For all $\beta>0$, the ML estimator for MGGD exists and unique if $\mathrm{sp}\{\bz_1,\bz_2,\cdots,\bz_n\}=\mathbb{R}^p$.
 \end{corollary}

 Existence and uniqueness of the MLE in MGGD has been previously addressed in \cite[Section V.A]{Ollila2012}. However, this contribution does not identify geodesic convexity and applies only to $0<\beta\leq 1$. Our theorem applies to MGGD with all $\beta>0$ and provides additional insight based on geodesic convexity. Furthermore, it applies to other ES distributions, including cases where $\rho'(x)$ is not continuous, e.g., $\rho(x)=x$ when $x\leq 1$ and $\rho(x)=3x-2$ when $x> 1$.

%Another important contribution of this theorem is that, the proof based on geodesic convexity is more general and can be adapted to some straightforward extensions, such as ML estimators with regularization or under the assumption that the covariance is a Kronecker product (the special case where $\rho(x)=\log(x)/p$ as been presented in~\cite{Chen2011} and~\cite{Wiesel2012LSE}).

Here we remark that although ML estimator for elliptically symmetric (ES) distribution is the motivation of the argument, Theorem~\ref{thm:mainn} (and Theorem~\ref{thm:banded} in next section) hold even if $\rho(x)$ does not relate to an elliptically symmetric (ES) distribution. For example, Tyler's M-estimator can be written in the form of~\eqref{eq:ces_def} with $\rho(x)=\log(x)/2p$; however this $\rho(x)$ corresponds to the central angular Gaussian which is not a member in ES class~\cite[(36)]{Ollila2012}. Another example is in~\cite[page 5610, Example 1]{Ollila2012}, where $\rho'$ is set to be a huber function.

We remark that the  Theorem \ref{thm:mainn} also applied to the complex  elliptically symmetric (CES) distributions defined by
\begin{equation}
f(\bz)=C_{p,g}|\bSigma|^{-1}g\(\(\bz-\mu\)^H\bSigma^{-1}\(\bz-\mu\)\),
\end{equation}
where the MLE estimator minimizes
\[
L_0(\bSigma)=\sum_{i=1}^{n}\rho\(\bz_i^H\bSigma^{-1}\bz_i\)+n\log\det(\bSigma).
\]
Then Theorem \ref{thm:mainn} holds with  $a_1=\sup\{a|x^{a}\exp(-\rho(x))\rightarrow 0\,\,\,\, \text{as $x\rightarrow \infty$}\}$. This is also a generalization of the uniqueness/existence of the minimizer of $L_0(\bSigma)$ presented in~\cite[Section V.A]{Ollila2012}.

 \begin{proof}
(a) Applying ~\cite[Proposition 1]{Wiesel2012}, if $\bSigma_3$ is the geodesic mean of $\bSigma_1$ and $\bSigma_2$ defined in~\eqref{eq:geodesic}, then
\begin{equation}\label{eq:convexity30}
\ln\(\bz^T\bSigma_1^{-1}\bz\)+\ln\(\bz^T\bSigma_2^{-1}\bz\)\geq 2\ln\(\bz^T\bSigma_3^{-1}\bz\).
\end{equation}
Combining this fact with the convex/monotone properties of $\rho$, we have
\begin{align}
&\rho\(\bz_i^T\bSigma_1^{-1}\bz_i\)+\rho\(\bz_i^T\bSigma_2^{-1}\bz_i\)
\nonumber\\\geq& 2\rho\(\exp\(\(\ln\(\bz_i^T\bSigma_1^{-1}\bz_i\)+\ln\(\bz_i^T\bSigma_2^{-1}\bz_i\)\)/2\)\)
\label{eq:convexity31}\\\geq &2\rho\(\exp\(\ln\(\bz_i^T\bSigma_3^{-1}\bz_i\)\)\)
= 2\rho\(\bz_i^T\bSigma_3^{-1}\bz_i\),\label{eq:convexity3}
\end{align}
Since
\begin{eqnarray}
 \log\det\(\bSigma_1\)+\log\det\(\bSigma_2\)=2\log\det\(\bSigma_3\),
\end{eqnarray}
we proved
\begin{eqnarray}
 L_0\(\bSigma_1\)+L_0\(\bSigma_2\)\geq 2L_0\(\bSigma_3\)
\end{eqnarray}
By Proposition~\ref{prop:mean}, we proved the geodesic convexity of $L_0(\bSigma)$.

 When $\rho(e^x)$ is strictly convex, $L_0(\bSigma)$ is geodesically strictly convex since the equalities in \eqref{eq:convexity31} and \eqref{eq:convexity3} can not hold simultaneously for all $1\leq i\leq n$. The proof is as follows. Following the proof of~\cite[Theorem III.1]{zhang2012robust}, when $\mathrm{sp}\{\bz_1,\bz_2,\cdots,\bz_n\}=\mathbb{R}^p$, the equality in \eqref{eq:convexity30} (and therefore the equality in~\eqref{eq:convexity3}) holds for all $1\leq i\leq n$ only when $\bSigma_1=c\bSigma_2$ for some $c\neq 1$. However, $\bSigma_1=c\bSigma_2$ would fail the equality in \eqref{eq:convexity31} due to the strict convexity of $\rho(e^x)$.
%(b) Applying the same argument in~\cite[Lemma 2.2(i)(ii)]{Kent1991}, $L_0(\bSigma)\rightarrow\infty$ when $\bSigma$ goes to the boundary of $\PSD(p)$, i.e., an eigenvalue of $\bSigma$ goes to $0$ or $\infty$. Therefore the minimizer of $L_0(\bSigma)$ lies inside $\PSD(p)$.
 \end{proof}

In practice, various numerical techniques can be used to find a local minima of the MGGD negative log-likelihood.  Theorem \ref{thm:mainn} and geodesic convexity then ensure that this local minima will also be global minima. A promising approach is the classical iterative reweighed scheme due to~\cite{Kent1991,Arslan2004}:
\begin{equation}\label{eq:mestimator_covariance}
\bSigma_{m+1}=\sum_{i=1}^{n}u(\bz_i^T\bSigma_m^{-1}\bz_i)\bz_i\bz_i^T/n,
\end{equation}
where $u(x)=\rho'(x)$. It has been shown in~\cite[Proposition 1]{Arslan2004} that when $\mathrm{Sp}(\bz_1,\bz_2,\cdots,\bz_n)=\mathbb{R}^p$, $\rho''(x)$ is continuous and nonnegative, then $L_0(\bSigma_m)$ decreases monotonically ($\mathrm{Sp}(\bz_1,\bz_2,\cdots,\bz_n)=\mathbb{R}^p$ is assumed in order to make sure that $L_0(\bSigma_{m+1})$ is well-defined). \cite[Proposition 3]{Arslan2004} shows that any limiting point of the sequence $\bSigma_m$ is a stationary point. When the assumptions in Theorem~\ref{thm:mainn} hold, this point is the unique minimizer of $L_0(\bSigma)$.

%\begin{equation}
%\bSigma=\sum_{i=1}^{n}u\(\bz_i^T\bSigma^{-1}\bz_i\)\bz_i\bz_i^T/n.
%\end{equation}
%When $L_0$ is strict convexity (for example, assumptions in Theorem~\ref{thm:mainn} are satisfied), then such stationary point is unique, since $\bSigma-\sum_{i=1}^{n}u(\bz_i^T\bSigma^{-1}\bz_i)\bz_i\bz_i^T/n$ is the derivative of the objective function $L_0(\bSigma)$ with respect to $\bSigma^{-1}$. Therefore $\bSigma_m$ converges to the unique minimizer of $L_0(\bSigma)$.

\section{Convexity in chordal MGGD}\label{convex_mggd}
In this section, we consider structured MGGD estimation. In particular, we consider MLE of the MGGD scatter matrix subject to sparsity constraints. Thus, we are interested in the solution to
\begin{eqnarray}\label{eq:ggm}
\begin{array}{ll}
  \min_{\bSigma} & \sum_{i=1}^{n}\rho\(\bz_i^T\bSigma^{-1}\bz_i\)+n\log\det(\bSigma)\\
 {\mathrm{s.t.}} & \left[\bSigma^{-1}\right]_{ij}=0,\quad (i,j)\notin E
\end{array}
\end{eqnarray}
where the objective is the negative MGGD log-likelihood as described in Section \ref{gcovnex_mggd}, and $E$ is the edge set of a known graph $G(V,E)$ associated with the sparsity of $\bSigma^{-1}$. Unfortunately, the above minimization is not convex in $\bSigma$ or $\bSigma^{-1}$. Nor is it geodesically convex in it. Indeed, the sparsity constraints are not preserved by the positive definite geodesic in (\ref{eq:geodesic}). Thus, it is not clear whether its global minima can be found in an efficient manner.

In what follows, we propose a simple trick to ``convexify'' the optimization in many interesting cases of MGGD. In particular, we assume the $G$ is chordal and represent $\bSigma^{-1}$ with its Cholesky factor $\bC$. Due to Proposition \ref{prop:chol}, a unique chordal decomposition exists and we obtain
\begin{eqnarray}\label{chol-opt}
\begin{array}{ll}
  \min_{\bC\in \sY} & \sum_{i=1}^{n}\rho\(\|\bC^T\bz_i\|^2\)-2n\sum_{j=1}^p\log(\bC_{j,j})\\
 {\mathrm{s.t.}} & \left[\bC\right]_{ij}=0,\quad (i,j)\notin E,
\end{array}
\end{eqnarray}
where
\begin{eqnarray}
\sY=\left\{\bC\in\mathbb{R}^{p\times p}:
\begin{array}
 {l}\bC_{i,i}>0 \\
 \text{$\bC_{i,j}=0$ for all $i<j$}
\end{array}
\right\}
\end{eqnarray}
and we have used
\begin{eqnarray}
 n\log\det\(\(\bC\bC^T\)^{-1}\)=-2n\sum_{j=1}^p\log(\bC_{j,j}).
\end{eqnarray}
 The following theorem characterizes the properties of this optimization.
 \begin{theorem}\label{thm:banded}
 (a) Assume that $\rho(x)$ is continuous in $(0,\infty)$, nondecreasing and $\rho\(x^2\)$ is convex, the objective of (\ref{chol-opt})  is strictly convex.
% When $\mathrm{sp}\{\bz_1,\bz_2,\cdots,\bz_n\}=\mathbb{C}^p$ and $\rho(x^2)$ is strictly convex, $L_0(\bSigma)$ is strictly convex with respect to $\bC\in\sY$.
 (b) When $\mathrm{sp}\{\bz_1,\bz_2,\cdots,\bz_n\}=\mathbb{R}^p$,  a minimizer to (\ref{chol-opt}) exists.% a minimizer of  $L_0(\bSigma)$ in $\PSD(p)$.
\end{theorem}

\begin{proof}
Negative logarithms are strictly convex, and it remains to show that
\begin{align}
&\rho\(\left\|\(\frac{\bC_1+\bC_2}{2}\)^T\bz_i\right\|^2\)\nonumber\\
&\qquad\leq\rho\(\(\frac{\|\bC_1^T\bz_i\|+\|\bC_2^T\bz_i\|}{2}\)^2\)\nonumber\\
&\qquad\leq\frac{1}{2}\rho\(\|\bC_1^T\bz_i\|^2\)+\frac{1}{2}\rho\(\|\bC_2^T\bz_i\|^2\)\end{align}
where we have used the triangle inequality and the convexity of $\rho\(x^2\)$.

%Applying proposition~\ref{prop:chol}, we rewrite the objective function $L_0(\bSigma)$ as a function depending on $\bC$, where $\bC$ is the Cholesky factor of $\bSigma^{-1}$:
%\begin{align*}
%L_1(\bC)=&\sum_{i=1}^{n}\rho(\bz_i^T\bC\bC^T\bz_i)+{n}\log\det((\bC\bC^T)^{-1})
%\\=&\sum_{i=1}^{n}\rho(\|\bC^T\bz_i\|^2)-2n\sum_{i=1}^p\log[\bC]_{ii},
%%=2\sum_{i=1}^{n}\rho(\|\bC^T\bz_i\|^2)-n\log\prod_{i=1}^p\bC_{i,i},
%\end{align*}
%and we minimize $L_0$ over the convex set $\sY$.
%
%Applying the assumptions on $\rho$, we prove the convexity $L_0$ by verifying $L_1(\bC_1)+L_1(\bC_2)\geq L_1(\bC_3)$, where $\bC_3=(\bC_1+\bC_2)/2$:
%\begin{align*}
%&\rho\(\|\bC_1^T\bz_i\|^2\)+\rho\(\|\bC_2^T\bz_i\|^2\)-2 \rho\(\|\bC_3^T\bz_i\|^2\)
%\\\geq &\rho\(\|\bC_1^T\bz_i\|^2\)+\rho\(\|\bC_2^T\bz_i\|^2\)-2 \rho\(\(\frac{\|\bC_1^T\bz_i\|+\|\bC_2^T\bz_i\|}{2}\)^2\)
%\\\geq& 0,
%\end{align*}
%and the strict concavity of $\sum_{i=1}^p\log[\bC]_{ii}$.

To prove the existence of minimizer, we need to prove that for $L_1(\bC)=\sum_{i=1}^{n}\rho\(\|\bC^T\bz_i\|^2\)-2n\sum_{j=1}^p\log(\bC_{j,j})$, $L_1(\bC)\rightarrow\infty$ as $\|\bC\|_F\rightarrow\infty$ or   the smallest eigenvalue of $\bC$ goes to $0$. When $\|\bC\|_F\rightarrow\infty$, since  $\mathrm{sp}\{\bz_1,\bz_2,\cdots,\bz_n\}=\mathbb{R}^p$ and $\bC$ is full rank, there exists $c_2$ such that $\sum_{i=1}^{n}\|\bC^T\bz_i\|^2 \geq c_2 \|\bC\|_F^2$. Consider that $\rho(x^2)$ is convex, there exists a constant $c_1$ and $C_1$ such that when $\frac{1}{n}\sum_{i=1}^{n}\|\bC^T\bz_i\|^2>C_1$,
\begin{align*}
&\sum_{i=1}^{n}\rho(\|\bC^T\bz_i\|^2)\geq n\rho(\frac{1}{n}\sum_{i=1}^{n}\|\bC^T\bz_i\|^2)
\\\geq& n c_1 \sqrt{\frac{1}{n}\sum_{i=1}^{n}\|\bC^T\bz_i\|^2}
\geq \sqrt{n} c_1 c_2\|\bC\|_F.
\end{align*}
Therefore as $\|\bC\|_F\rightarrow\infty$, $L_1(\bC)\rightarrow\infty$.

 When the smallest eigenvalue of $\bC$ goes to $0$, we have $\det(\bC)\rightarrow 0$ and $\sum_{j=1}^p\log(\bC_{j,j})=\log\det(\bC)\rightarrow -\infty$. In another aspect, since $\rho(x^2)$ is convex, $\lim\inf_{x\rightarrow 0}\rho(x)$ exists (by convexity it is larger than $2\rho(1)-\rho(4)$). Combining it with the fact that $\rho(x)$ is nondecreasing, $L_1(\bC)\rightarrow\infty$ as the smallest eigenvalue of $\bC$ goes to $0$,
 and the existence of the minimizer of $L_1(\bC)$ is proved.
\end{proof}

The theorem shows that, under suitable conditions, any local minimizer of (\ref{chol-opt}) is globally optimal. %Applying the one-to-one correspondence between a positive definite matrix $\bSigma^{-1}$ and its Cholesky factor $\bC\in\sY$ (see Proposition~\ref{prop:chol}), any local minimizer of $L_0(\bSigma)$ is a global minimizer of $L_0(\bSigma)$.
This holds for any ES distribution with $\rho$ satisfying the assumptions. In particular, an important special case is chordal graphical models in MGGD: \begin{corollary}
If $\mathrm{sp}\{\bz_1,\bz_2,\cdots,\bz_n\}=\mathbb{R}^p$, under chordal graphical models the ML estimator for MGGD exists and unique for all $\beta\geq 0.5$.
 \end{corollary}

 Unfortunately, the recent result in \cite{abramovich2012expected} on $\rho(x)=\log(x)/2p$ does not satisfy our  conditions and requires a different analysis. And there is no detailed proof to support the claim in~\cite{abramovich2012expected} that this banded Tyler estimators converge to unique fixed points.

%We remark that when $\rho(x^2)$ is convex in Theorem~\ref{thm:banded} (not strictly convex), we can still obtain the strict convexity of $L_1(\bC)$, if for any $(p-1)$-dimensional subspaces $L_1,L_2\subset\mathbb{C}^p$:
%\begin{align}\label{eq:convex_condition}&\{\sX\cap\rmV_1\}\cup\{\sX\cap\rmV_2\}\neq\sX.
%\end{align}
% The proof can be found in~\cite[Theorem 1]{Zhang2011}.
%
% In this section we assume that the sparsity structure $G(V,E)$ is given. In practice there are some algorithms for estimating this structure, such as

When the condition in Theorem~\ref{thm:banded} holds, the objective function is convex with respect to $\bC$. Therefore, it can be numerically minimized via any general convex optimization solver. For example, in the MGGD case with $\beta\geq 0.5$ the problem can be expressed as
 \begin{eqnarray}
\begin{array}{ll}
  \min_{\bC\in \sY,\bf{t}} & \sum_{i=1}^{n}|t_i|^{2\beta}-2n\sum_{j=1}^p\log(\bC_{j,j})\\
 {\mathrm{s.t.}} & t_i\geq \|\bC\bz_i\|\\
 & \left[\bC\right]_{ij}=0,\quad (i,j)\notin E
\end{array}
\end{eqnarray}
where $\bf{t}$ are auxiliary variables. This formulation with second order cone constraints can be easily solved using the popular CVX package \cite{grant2008cvx}.

Alternatively, the optimization can be addressed using a majorization - minimization (MM) technique, e.g. \cite{Arslan2004}. The method begins with an initial estimate $\hat\bSigma_0\in\PSD(p)$ and updates it according the following iterations
\begin{eqnarray}\label{iterat}
 \hat\bSigma_{m+1}=\arg \left\{
\begin{array}{ll}
  \min_{\bSigma} & h(\bSigma,\bSigma_m)\\
 {\mathrm{s.t.}} & \left[\bSigma^{-1}\right]_{ij}=0,\quad (i,j)\notin E
\end{array}\right.
\end{eqnarray}
where
\begin{equation}\label{eq:auxillary}
h(\bSigma,\bSigma_m)=\sum_{i=1}^{n}u(\bz_i^T\bSigma_m^{-1}\bz_i) \bz_i^T\bSigma^{-1}\bz_i+n\log\det(\bSigma)+C,
\end{equation}
$u(x)=\rho'(x)$ and $C$ is chosen such that ${h(\bSigma_m,\bSigma_{m})= L_0(\bSigma_m)}$. Since $\rho''(x)$ is continuous and nonnegative, we have
\begin{eqnarray}
 L_0(\bSigma_{m+1})\geq h(\bSigma,\bSigma_{m})\geq h(\bSigma_m,\bSigma_{m})=L_0(\bSigma_m)
\end{eqnarray}
and $L_0(\bSigma_m)$ converges as $m\rightarrow\infty$. Each iteration step in (\ref{iterat})-(\ref{eq:auxillary}) can be interpreted as a Gaussian graphical model optimization (the weights $u(\cdot)$ are constant with respect to the optimization variable). These minimizations have a simple closed form solution when $G(V,E)$ is chordal (see appendix). Thus, the proposed technique is very efficient for implementation in practice.

Finally, it is worth mentioning that our framework can also be extended to structure learning in MGGD graphical models. Structure learning, also known as covariance selection, considers the estimation of an inverse covariance which is known to be sparse but the sparsity pattern itself is unknown \cite{dempster1972covariance,Bickel2008,Rothman2010}. Adding a sparsity constraint usually destroys the convexity. Instead, the modern approach relies on a convex relaxation based on an L1 norm penalty. In our context, the problem is convex in the Cholesky factor rather than in the inverse covariance itself. This leads to the following convex minimization
 \begin{eqnarray}
  \min_{\bC\in \sY}  \sum_{i=1}^{n}\rho(\|\bC\bz_i\|^2)-2n\sum_{j=1}^p\log([\bC]_{j,j})+\lambda\|\bC\|_1\end{eqnarray}
where $\lambda$ is a regularization parameter, and $\|\bC\|_1$ is a matrix version of the L1 norm, namely a sum over the absolute values of the elements in $\bC$. It is important to emphasize that this approach is only applicable to chordal graphical models and it assumes that the perfect order of the variables is known a priori. Recent developments in high dimensional covariance estimation provide data-driven methods for identifying this order and structure \cite{rutimann2009high,rolfs2012natural}. We leave this topic as a possible direction for future research.

Furthermore, our methodology can be easily extended to joint estimation of the covariance and the centering parameter $\mu$. For this purpose, consider the optimization of
\begin{eqnarray}
\begin{array}{ll}
  \min_{\bC\in \sY, \mu} & \sum_{i=1}^{n}\rho\(\|\bC^T\bz_i-\bx\|^2\)-2n\sum_{j=1}^p\log(\bC_{j,j})\\
 {\mathrm{s.t.}} & \left[\bC\right]_{ij}=0,\quad (i,j)\notin E
\end{array}\!\!\!\!\!\!\!\!\!\!\!\!\!\!\!\!\nonumber\\
\end{eqnarray}
where $\bx=\bC^T{\bm{\mu}}$
When the conditions in Theorem~\ref{thm:banded}(a) hold, the objective function is jointly convex with respect to $(\bC, \bx)$, therefore any of its local minimizer is also its global minimizer. Its algorithm can be similarly addressed using a majorization - minimization (MM) technique as in~\eqref{iterat}-\eqref{eq:auxillary}.
%\section{Other examples of ML estimation}
%Though in this paper we mainly consider MGGD due to its popularity, Theorem~\ref{thm:mainn} can also be applied to other distributions. For example, for multivariate t-distribution, $\rho(x)=\frac{\nu +p}{2}\log(1+x/\nu)$, and conditions in Theorem~\ref{thm:mainn}(a) are satisfied while $a_1=\nu+p$ for Theorem~\ref{thm:mainn}(b). However this distribution fails to meet the assumptions in Theorem~\ref{thm:banded}.
%
%For Weibull distribution where $g(x)=x^{s-1}\exp(-x^s/b)$ and $\rho(x)=x^s/b-(s-1)\log x$, the conditions in Theorem~\ref{thm:mainn}(a) are satisfied for $0<s<1$, while condition in Theorem~\ref{thm:mainn}(b) holds since $a_1=+\infty$. However, the condition in Theorem~\ref{thm:banded} are not satisfied.
%
%For Complex Weibull distribution, where $\rho(x)=-\frac{\nu-m}{2}\log t-\log K_{\nu-m}(2\sqrt{\nu t})$.
%%$+\sqrt{\frac{\nu}{t}}K'_{\nu-m}(2\sqrt{\nu t})$.

\section{Numerical results}\label{simulations}
In this section, we present results of numerical simulations. The purpose of these simulations is three fold: to validate our theoretical results using synthetic data, to provide insight on a few open theoretical questions using synthetic data and to demonstrate the usefulness of our theoretical models in a real world setting.

The first experiment considered inverse covariance estimation using synthetic data generated from a known MGGD distribution with $\beta\geq 0.5$ and a chordal structure. The theory in Section \ref{convex_mggd} shows that, in this case, the MLE can be formulated as a convex optimization. To verify this result, we generated $n$ i.i.d. realizations of an MGGD with $p=10$, $\beta=0.5$ (corresponding to a classical multivariate Laplacian distribution), and a banded inverse covariance of width $b=4$. In particular, we used a Toeplitz inverse covariance with $1.0$ diagonal elements and $0.4$ off-diagonal elements within the main band. We estimated the unknown covariance using five estimators:
\begin{itemize}
 \item G: A Gaussian MLE with no prior knowledge of the structure, corresponding to the classical sample covariance.
 \item BG: A Gaussian banded MLE as detailed in the appendix.
 \item MGGD: An MGGD MLE with no prior knowledge of the structure via $30$ iterations of (\ref{eq:mestimator_covariance}).
 \item BMGGD1: An MGGD banded MLE via $30$ iterations of (\ref{iterat}) and the subroutine in the appendix. This estimator is initialized with $\hat\bSigma=\bI$.
 \item BMGGD2: An MGGD banded MLE via $30$ iterations of (\ref{iterat}) and the subroutine in the appendix. This estimator is initialized with $\hat\bSigma=\bSigma$ (which is clearly impossible in practice).
\end{itemize}
In all estimators above we use $30$ as the number of iterations since in our simulations the fixed point algorithms~\eqref{eq:mestimator_covariance} and \eqref{iterat} converge well after $30$ iterations.

Figure 1 shows the normalized mean squared Frobenius error in the covariance averaged over $10000$ independent simulations as a function of the number of samples $n$. It is easy to see the performance advantage of banded MGGD estimator. As expected, BMGGD1 and BMGGD2 converge to the identical fixed points irrespective of their initial condition.

The second experiment is very similar to the first except that $\beta=0.2$. Our theory assumes $\beta\geq 0.5$ and therefore does not hold for this value of shape parameter. Yet, according to \cite{abramovich2012expected} banded Tyler estimators do converge to unique fixed points, and these can be interpreted as the limit of MGGD when $\beta\rightarrow 0$. Our proposed fixed point iteration in (\ref{iterat}) can be implemented with a small $\beta$. Therefore, it is interesting to examine its performance and we repeat the first experiment with $\beta=0.2$. The results are presented in Figure 2. Interestingly, BMGGD1 and BMGGD2 converged to identical fixed points irrespective of their initial condition. Thus, although we have no proof for this behavior, we conjecture that, in practice, the iteration can also be used for all MGGDs.

The third experiment focused on non-chordal structures. Our theory assumes a chordal sparsity pattern and it is interesting to see whether this assumption is indeed critical. Thus, we repeated the first experiment associated with MGGD $\beta=0.5$,  but replaced the banded inverse covariance with a loopy structure. In particular, we constructed a two dimensional grid graph of size $p=3^2$ with the edges $(3,1)$, $(5,1)$, $(6,1)$, $(7,1)$, $(8,1)$,  $(9,1)$,  $(4,2)$,  $(6,2)$,  $(7,2)$,  $(8,2)$,  $(9,2)$,
  $(4,3)$,  $(5,3)$,   $(7,3)$,  $(8,3)$,   $(9,3)$,  $(6,4)$,   $(8,4)$,  $(9,4)$,   $(7,5)$,  $(9,5)$,
    $(7,6)$,  $(8,6)$, and  $(9,7)$. All the non-zero off diagonal elements were assigned a constant value small enough to ensure that $\bSigma^{-1}$ will be well conditioned. In contrast to the chordal case, MLE in general Gaussian Graphical models in does not satisfy a closed form solution. Instead, we solved  (\ref{eq:ggm}) using the CVX  optimization package \cite{grant2008cvx}. The latter was used for implementing BG, and for the inner solution in each iteration of BMGGD1 and BMGGD2. The results averaged over $400$ simulations\footnote{Due to CVX these simulations are highly time consuming.} are provided in Figure 3. Here too, the iterations converged to identical solutions and seem to be independent of their initial conditions.

The fourth experiment addressed the practical use of the BMGGD1 estimator in a real world example. Following \cite{Rothman2010} we considered the SONAR dataset from the UCI machine learning data repository. This dataset has 111 spectra from metal cylinders and 97 spectra from rocks, where each spectrum has 60 frequency band energy measurements. Quadratic discriminant test is used to classify the metals and the rocks. It requires the estimation of the covariance in both classes, and previous work demonstrated the advantage of BG over the classical sample covariance and the naive diagonal estimate. We repeated the experiment step by step and added BMGGD1. Specifically, we chose $\beta$ as the parameter within $\{0.5,0.6,\cdots,0.9,1.0\}$ that maximizes the MGGD likelihood, and used the same band which was used by BG (selected via 10 random splits with $1/3$ of the data for training and the validation likelihood). For the QDA test we applied the covariance of MGGD, which is $c(\beta)\bSigma$, where  $c(\beta)=2^{\frac{1}{\beta}}\frac{\Gamma(\frac{p+2}{2\beta})}{p\Gamma(\frac{p}{2\beta})}$.  The test errors over a standard leave-one-out cross validation are provided in Table I.  These results remained stable over different randomizations, and demonstrate the advantage of the proposed BMGGD1 framework.

\begin{table}[h]
\caption{Test errors in SONAR dataset.}
\begin{center}
\begin{tabular}{|c|c|c|c|}
\hline
Sample covariance & Naive Bayes & BG & BMGGD1 \\
\hline
$24.0\%$ &  $32.7\%$ & $15.4\%$ & $13.5\%$\\
\hline
\end{tabular}
\end{center}
\end{table}%

\begin{figure}
\center
\includegraphics[width=.5\textwidth]{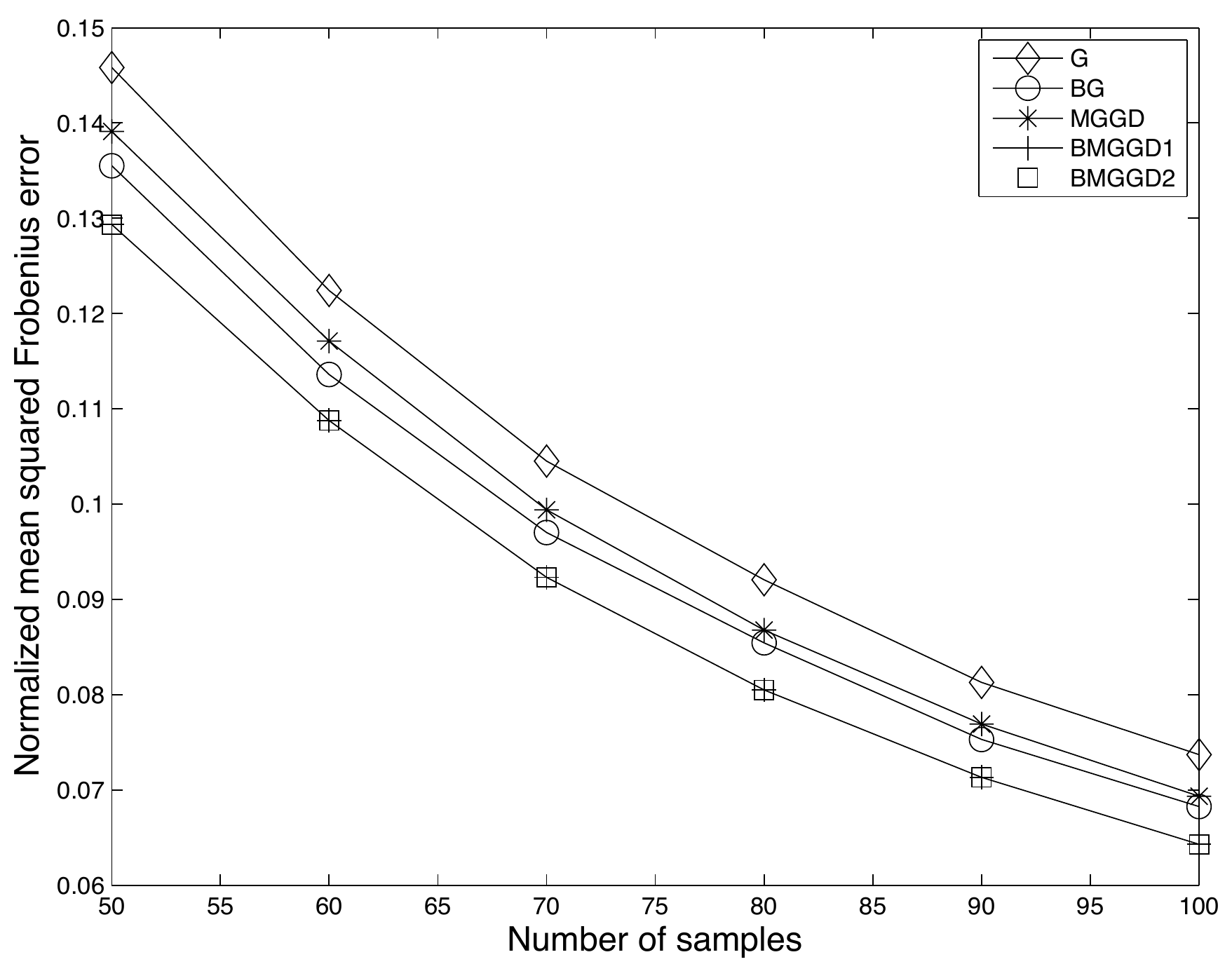}
\caption{MGGD with $\beta=0.5$ and a banded inverse covariance .}\label{fig1}
\end{figure}

\begin{figure}
\center
\includegraphics[width=.5\textwidth]{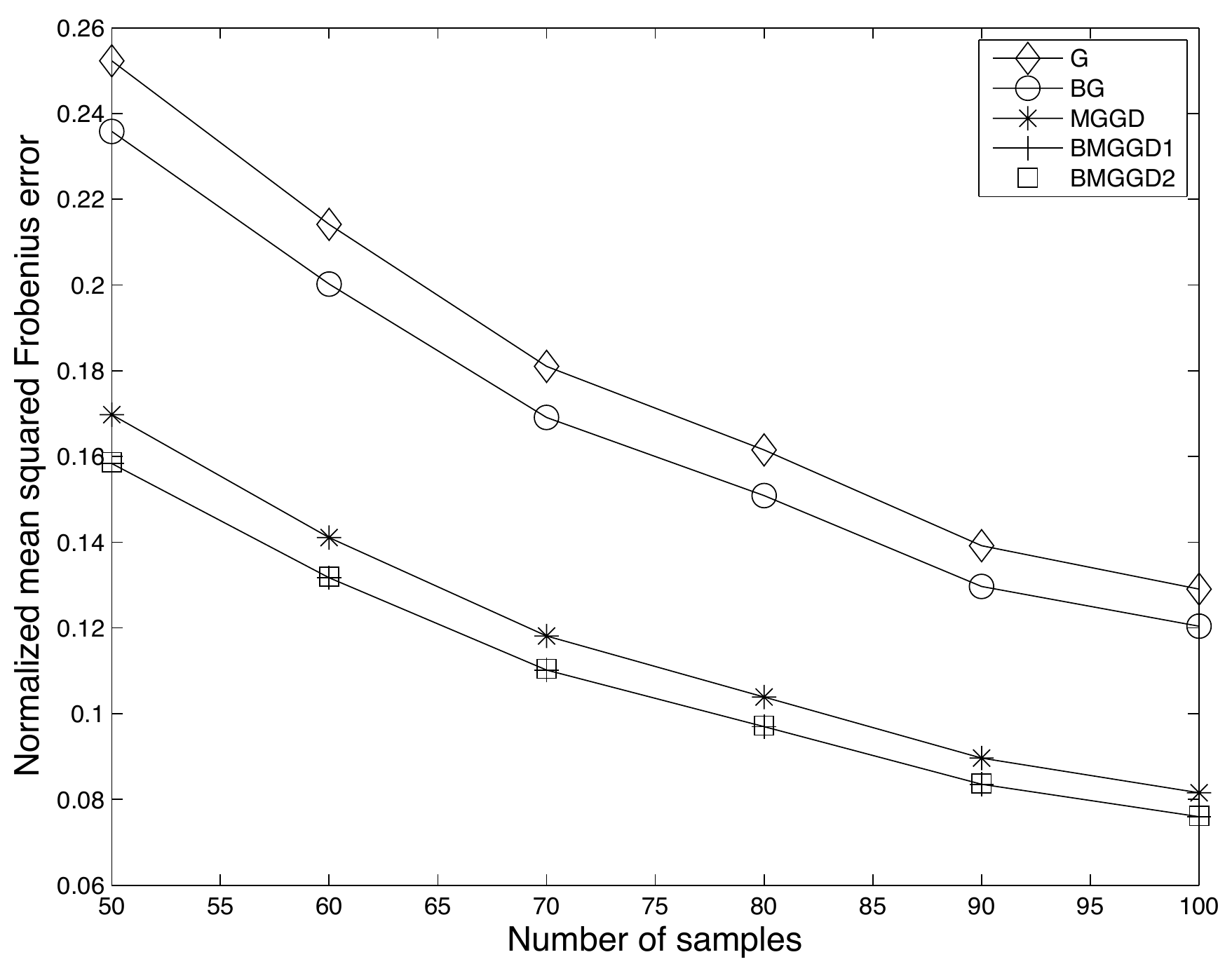}
\caption{MGGD with $\beta=0.2$ and a banded inverse covariance .}\label{fig2}
\end{figure}

\begin{figure}
\center
\includegraphics[width=.5\textwidth]{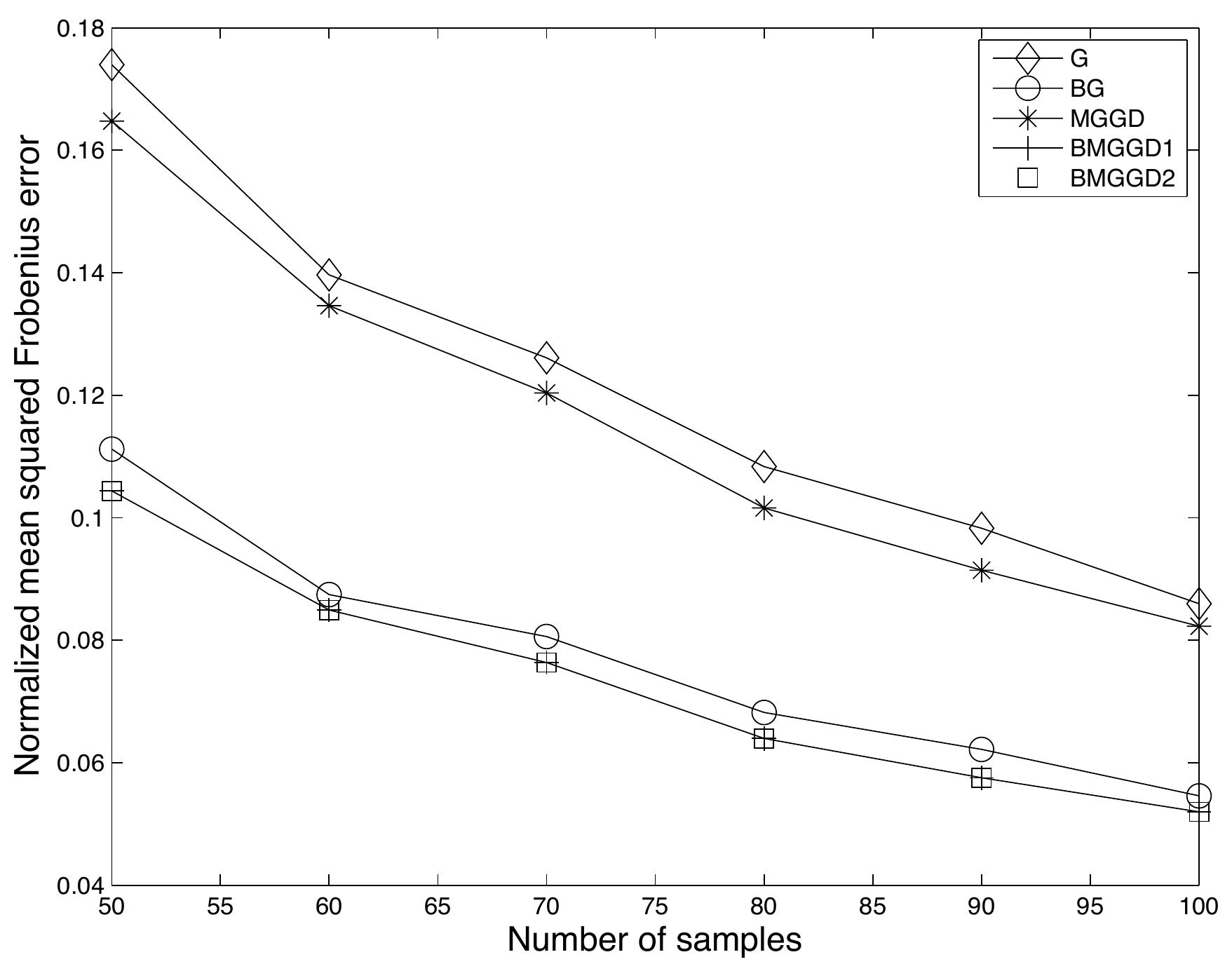}
\caption{MGGD with $\beta=0.5$ and a non-chordal graphical model.}\label{grid}
\end{figure}

\section{Discussion}\label{conclusion}
In this paper, we consider covariance estimation in the multivariate
generalized Gaussian distribution (MGGD). We proved that the MGGD negative log-likelihood is geodesically convex for $\beta>0$. In the sparsely constrained case, we proved that a simple change of variables can transform it into a convex function as long as $\beta\geq0.5$ and the underlying graph is chordal. This means that any local solution of these minimization is globally optimal and the problems can be solved using standard descent methods. In practice, we observed this behavior also for smaller values of $\beta$. This agrees with a similar result on banded Tyler methods which can be interpreted as $\beta\rightarrow 0$. An interesting direction for future work is a rigorous analysis of these phenomenon when $0<\beta<0.5$. Another direction is relaxing the chordal assumption on the sparsity pattern. Here too, our numerical experience suggests that simple descent methods converge to the global solution and are independent of their initial conditions. Finally, as mentioned above, it in interesting to examine the problem of structure learning in MGGDs via sparsity enforcing priors.

\appendix
In this appendix, we review a simple closed form solution for the MLE in chordal (decomposable) Gaussian graphical models \cite{lauritzen1996graphical,Wiesel_decomposable2010,andersen1995linear}. The problem is
 \begin{eqnarray}\label{ggm}
\begin{array}{ll}
  \min_{\bSigma} & \sum_{i=1}^n\alpha_i\bz_i^T\bSigma^{-1}\bz_i+n\log\det{\bSigma}\\\
 {\mathrm{s.t.}} & \left[\bSigma^{-1}\right]_{ij}=0,\quad (i,j)\notin E
\end{array}
\end{eqnarray}
where
\begin{eqnarray}
 \alpha_i=u(\bz_i^T\bSigma_m^{-1}\bz_i)
\end{eqnarray}
are the iteration weights which do not depend on $\bSigma$. Using the chordal Cholesky decomposition $\bSigma^{-1}=\bC\bC^T$, the problem is equivalent to
 \begin{eqnarray}
\begin{array}{ll}
  \min_{\bC} & \sum_{i=1}^n\alpha_i\|\bC^T\bz_i\|^2 - 2 n \sum_{j=1}^p\log ([\bC]_{jj})\\\
 {\mathrm{s.t.}} & \left[\bC\right]_{ij}=0,\quad (i,j)\notin E \text{ or } i>j
\end{array}
\end{eqnarray}
This minimization is completely separable and each column of $\bC$ can be simply obtained by a linear regression. Let $\{\bZ_i\}_{i=1}^p \subset\mathbb{R}^n$ be vectors consists of the $i$-th components of $\sqrt{\alpha_1}\bz_1,\sqrt{\alpha_2}\bz_2,\cdots,\sqrt{\alpha_n}\bz_n$, $J_i=\{i+1\leq k\leq p: (i,k)\in E\}$, and assume that in the linear regression of  $\bZ_i$ with respect to $\{\bZ_j\}_{j\in J_i}$, the parameter on $\bZ_j$ is $\beta(i,j)$, and standard error is $r_i$. Then
\begin{eqnarray}
 [\bA]_{ij}&=&\begin{cases}1, &\text{when $i=j$} \\
 \beta(j,i), &\text{when $i>j$ and $(i,j)\in E$}\\
 0, &\text{otherwise}.
 \end{cases}\\
 \bD&=&\diag(r_1,r_2,\cdots,r_p),
\end{eqnarray}
 and finally $\bC=\bD\bA$.
%solution can be found by a column-by-column minimization over $\bC$. Let $\bZ_i (1\leq i\leq p)\in\mathbb{C}^n$ be a vector consists of the $i$-th components of $\alpha_1\bz_1,\alpha_2\bz_2,\cdots,\alpha_n\bz_n$. CHANGE TO GENERAL CHORDAL!!! Then the $i$-th column of $\bC$ can be obtained by the regression of $\bZ_i$ over $\{\bZ_j\}_{j\in J_i}$, where $J_i=\{i+1\leq k\leq p: (i,k)\in E\}$. Assuming that in the regression $\bZ_i$ over $\{\bZ_j\}_{j\in J_i}$, the parameter on $\bZ_j$ is $\beta(i,j)$, and standard error of the regression is $r_i$
%\begin{align*}\bA_{i,j}:& = -\beta(i,j)\,\,\, \text{when $i\in J_i$}
%\\
%\bA_{i,i}:&=1 \\
%\bA_{i,i}:&=0\,\,\, \text{when $j\neq i$ and $j\neq J_i$}.
%\end{align*}
%Let $\bD$ be a diagonal matrix where the $i$-th element being the standard error of the regression: $
%\bD_{i,i}=r_i.$
%Then the solution $\bC$ is $\bC=\bA\bD$.

\bibliographystyle{abbrv}
\bibliography{bib-rrp}
\end{document}